%% file: main-buffer-model.tex
\documentclass{article}

\usepackage{amsmath}
\usepackage{amssymb}
\usepackage{braket}
\usepackage{xspace}
\usepackage{fullpage}
\usepackage{color}

\usepackage[final]{changes}

\newcommand{\push}{{\sc Push}\xspace}
\newcommand{\pull}{{\sc Pull}\xspace}
\newcommand{\pp}{{\sc Push\&Pull}\xspace}

\newtheorem{theorem}{Theorem}
\newtheorem{lemma}{Lemma}

\newtheorem{corollary}{Corollary}

\newcommand{\qed}{\hfill $\square$ \smallbreak}
\newenvironment{proof}{\noindent\textbf{Proof:}}{\qed}

\begin{document}


%
\title{Randomized Rumor Spreading in Ad Hoc Networks with Buffers\thanks{This manuscript was submitted to DISC 2013.}}
%
%
\author{Dariusz R. Kowalski%
	\thanks{Department of Computer Science, University of Liverpool, Ashton Building, Ashton Street, Liverpool L69 3BX, UK.
		Supported by the Engineering and Physical Sciences Research Council  [grant number EP/G023018/1].} 
\and Christopher Thraves Caro
\thanks{GSyC, Universidad Rey Juan Carlos,
	Campus de Fuenlabrada,
	Camino del Molino~S/N,
	28943 Fuenlabrada, Madrid,
	Spain.
	Supported by Spanish MICINN grant Juan de la Cierva, Comunidad de Madrid grant S2009TIC-1692
and Spanish MICINN grant TIN2008--06735-C02-01.}
}
%
%
%

\date{21 May 2013}

\maketitle              

\begin{abstract}
The randomized rumor spreading problem generates a big interest in the area of distributed algorithms due to its simplicity, robustness and wide range of applications. 
The two most popular
communication paradigms used for spreading the rumor are \push and \pull algorithms. 
The former protocol allows nodes to send the rumor to a randomly selected neighbor at each step, 
while the latter is based on sending a request and downloading the rumor from a randomly selected
neighbor, provided the neighbor has it.
Previous analysis of these protocols assumed that every node could process all such push/pull
operations within a single step, which could be unrealistic in practical situations.
Therefore we propose a new framework for analysis rumor spreading
accommodating buffers, in which a node can process only one push/pull message or push request at a time.
We develop upper and lower bounds for randomized rumor spreading time in the new framework,
and compare the results with analogous in the old framework without buffers.
\\

\noindent
{\bf Keywords:} Randomized rumor spreading, model with buffers, synchronous model, asynchronous model, push and pull protocols.
\end{abstract}
%


\newpage

\input{intro.tex}

\input{rel-work.tex}

\input{model.tex}

\input{push.tex}

\input{pull.tex}

\input{conclusions.tex}

%


\end{document}

%% file: intro.tex
\section{Introduction}\label{sec:intro}

Rumor spreading is a fundamental concept in ad hoc communication, databases and systems:
a rumor, initially stored in one node, needs to be delivered to all other nodes
in the network by passing it along available links.
\emph{Randomized and distributed} spreading algorithms 
are of special interest due to their simplicity, robustness and locality 
\cite{Karp:2000:RRS:795666.796561,DBLP:journals/rsa/FeigePRU90}.
A well known randomized spreading algorithm proposed in the literature  \cite{Demers1988} is called \push. \push algorithm is executed in each node and it works in 
steps. At each step: if a node holds the rumor, it chooses 
uniformly at random one neighbor to establish a point-to-point communication. If the chosen neighbor does not hold the rumor, the rumor is \emph{pushed} through 
and the neighbor becomes a holder of the rumor at the same step. If a node does not hold the rumor, it waits until it becomes a holder. 
\push algorithm has 
its counterpart
known as \pull spreading algorithm. \pull  algorithm is executed in each node and it also works in 
steps. At each step: if a node does not hold the rumor, it chooses 
uniformly at random one neighbor to establish a point-to-point communication. If the chosen neighbor holds the rumor, the rumor is \emph{pulled} through 
and the node becomes a holder of the rumor at the same step. If a node holds the rumor, it waits until every node holds the rumor (while neighbors may \emph{pull} the rumor from it). 
A third well studied randomized algorithm is a combination of 
\push and \pull algorithms. It is called \pp and it works as follows: at each step, every node that holds the rumor acts according to \push algorithm, and 
every node that does not hold the rumor acts according to \pull algorithm.  
Some appealing applications of randomized rumor spreading are data aggregation \cite{1638541}, 
maintenance of replicated databases \cite{Demers1988},  
resource discovery \cite{Harchol-Balter1999} and failure detector \cite{vanRenesse:2009:GFD:1659232.1659238}, among others. 

The classical framework that was used for studying the behavior of \push, \pull and \pp protocols assumes discrete time running in synchronized time steps for all nodes. 
At each time step, every node can establish many parallel point-to-point communications if many neighbors request one via 
\push or
\pull. Hence, it is implicitly 
assumed that, at each step, every node is able to process every received 
\push message and
\pull request. Therefore, if a node that holds the rumor is contacted 
by many neighbors at the same time step trying to 
	\emph{push} or
\emph{pull} the rumor, the classical framework allows this node to 
process all the \push messages and/or
answer all the \pull requests in the same time step.   
%
Such an assumption makes that some important phenomena 
occurring in practical situations
may not be
captured. 
For instance, if the graph consists in one central node connected to $n$ leaves and the spreading algorithm is \pp, the results proved in 
\cite{DBLP:conf/stacs/Giakkoupis11} say that the rumor is spread in at most $O(\log n)$ time steps. While, even if the central node can send $n$ messages at the same time, 
it 
may need in practice
$n$ steps 
to process all the messages 
(\push messages and/or \pull requests)
sent by the $n$ leaves. Therefore, if the source node is one of the leaves, in the worst case, the central node 
may learn
the rumor 
after $n$ steps (after processing all the $n$ messages sent by the leaves). 
This
fact 
is clearly not captured by the analysis done 
so far in the previous work on randomized rumor spreading, which assumed
such parallelism when 
processing \push messages and
answering \pull requests. 

\paragraph{Our contributions.}
In order to address the aforementioned phenomena,
we introduce a new framework to study rumor spreading. 
In the literature, we find the \emph{phone call} model used to study gossiping problem \cite{DBLP:journals/siamcomp/KrummeCV92,DBLP:journals/dam/ChlebusDP94}. 
The model we introduce in this document can be seen as the phone call model where every node has a \emph{telephone answering machine} (the buffer). 
Therefore, if a node receives a call while it is already in a different call, the buffer takes the message. Then, the node  can check the buffer making a \emph{call} and it writes down 
the message in the local memory. In the model we introduce, a node is allowed to perform in one time step: a call to one neighbor, 
and
a call to the the buffer to take note of one message recorded in the buffer. 
If the buffer is empty at some step, one incoming call goes directly to the node and the rest incoming calls (if they exist) go to the buffer (the call that goes directly to the node is chosen uniformly at random). 
When the buffer is not empty, since the node is calling to the buffer to take note of a message, every incoming call goes directly to the buffer. 

We prove similarities and, most importantly, crucial differences 
between the performance of randomized spreading algorithms when they are analyzed under the new framework or under the classical framework for rumor spreading problem. 
In particular, when the spreading algorithm is \push, we prove that every result obtained in the classical model applies in the introduced model. 
Contrarily to what happens when the spreading algorithm is \pull: in that case, time complexity for rumor spreading increases considerably since the amount of messages stored in the buffer 
grows
rapidly. Particularly, we show that for some class that includes also Caterpillar graphs, i.e., a chain of star graphs, \pull algorithm 
requires at least $\Omega(\delta^D)$ rounds, where $D$ is the diameter of the graph and $\Delta_{max}$ is the degree of internal nodes. 
We also prove upper bounds for general graphs $O(\Delta_{max}\log\Delta_{max}(E_{max} - 1)^D)$ and for regular graphs $O(D^2 \Delta \log \Delta)$, 
where $E_{max} = \max_{v \in V}{\sum_{u \in N(v)} 1/d_u}$, $\Delta_{max}$ is the maximum degree and $\Delta$ is the degree in regular graphs. 




The rest of the document is organized as follows. We review the related work in Section \ref{sec:rel-work}.
In Section \ref{sec:model}, we define notation and give a detailed description of the introduced model.
In Section \ref{sec:push}, we prove the equivalence in terms of time complexity between the classical model and the model with buffers when \push algorithm is used to spread the rumor. 
In Section \ref{sec:pull}, we show crucial difference in terms of time complexity between the classical model and the model with buffers when \pull algorithm is used to spread the rumor. 
Finally, in Section \ref{sec:conclusions}, we provide conclusions and discuss about the future work.

%% file: rel-work.tex
\section{Related work}\label{sec:rel-work}
Time complexity for rumor spreading has been widely studied in the literature. 
Concerning rumor spreading via \push algorithm, time complexity has been studied for complete graphs in \cite{frieze1985shortest,Pittel:1987:SR:37387.37400}
proving that, in probability, the rumor is spread to all vertices within $\log_2n + \ln n + O(1)$ steps. 
Time complexity via \push algorithm in Cayley graphs was studied in \cite{DBLP:conf/stacs/ElsasserS07}. The authors proved a time complexity of $O(\log n)$ to Star graphs, Pancacke graphs and Transposition graphs, among others. Interestingly, they upper bounded the time complexity when \push algorithm spreads the rumor by the mixing time of a random walk in the graph. On the other hand, in \cite{fountoulakis2012rumor}, time complexity for \push protocol is studied for random regular graphs and expanders. Upper bounds for general graphs have been presented in \cite{DBLP:journals/rsa/FeigePRU90} and improved in \cite{DBLP:journals/tcs/ElsasserS09}. Time complexity for rumor spreading via \push algorithm has been studied as well in the asynchronous model in \cite{springerlink:10.1007/s00453-008-9245-4}.

When \pp algorithm is used to spread the rumor, time complexity has been upper bounded via vertex expansion when the graph is regular \cite{DBLP:conf/soda/SauerwaldS11} as well as for general graphs \cite{DBLP:conf/soda/GiakkoupisS12} where it was proved an upper bound of $O(\alpha^{-1}\log^{2.5})$, where $\alpha$ is the vertex expansion of the graph under analysis. On the other hand, time complexity for rumor spreading has been also studied via conductance for complete graphs   \cite{Demers1988,Karp:2000:RRS:795666.796561,Elsasser:2006:CCR:1148109.1148135,DBLP:conf/icalp/DoerrF11}, social networks \cite{DBLP:journals/tcs/ChierichettiLP11}, or general graphs \cite{DBLP:conf/soda/ChierichettiLP10,DBLP:conf/stoc/ChierichettiLP10,DBLP:conf/stacs/Giakkoupis11} where a tight upper bound of $O(\Phi^{-1} \log n)$ was proved, when $\Phi$ denotes the conductance of the graph under analysis. 
In the same work, an upper bound of $O(\Phi^{-1}(1 +\frac{\Delta_{max}}{\Delta_{min}})\log n)$ was proved for \pull algorithm in general graphs. 

Randomized spreading algorithms has slightly modified to substantially improve time complexity for tumor spreading. The authors of  
\cite{DBLP:conf/soda/Censor-HillelS11} used weak conductance to decrease form polynomial to polylogarithmic time complexity in graphs with large weak conductance by adding some determinism in the algorithm to discover bottlenecks. Finally, it was proved in \cite{DBLP:conf/stoc/Censor-HillelHKM12} that a simple modification of the \pp algorithm gives a solution for rumor spreading with time complexity of at most $O(D + \mbox{poly}\log n)$, where $D$ is the diameter of the network, hence, with no dependence on the conductance.  

%% file: model.tex
\section{Model}\label{sec:model}

The network is modeled by an undirected graph, denoted by $G=(V,E)$, 
with the set of $n$ nodes $V$, representing the computing entities, 
and the set of edges $E$, representing point-to-point bi-directional communication links available. 
Let $N(v)$ denote the set of neighbors of node $v$, i.e., the point-to-point bi-directional communication links available incident to node $v$. 
Let $D$ denote the diameter of the graph and $\Delta$ denote the maximum degree of the graph. 
We do not assume any global node or link labeling; instead,
only local link labeling is required to be able to select an outgoing port for a message to be sent
and to identify the incoming point of a received message.
Time is considered to be slotted in {\em synchronized steps}, also called {\em rounds}.

\paragraph{Buffers.}

Every node has a bounded size buffer. We use letter $B$ to denote the size of the buffer. 
$B$ is a parameter of the system, therefore, we might consider $B$ to be large enough so that it can be considered as ``unbounded" from perspective of a finite-time execution. 
In this document, indeed, we study the case when the size of the buffers is ``unbounded".
All messages sent to a node are queued in its buffer. 
Unless stated otherwise, messages are stored according to the FIFO (First-In-First-Out) policy, 
with ties broken either arbitrarily (by some adversary) or uniformly at random. 


\paragraph{Local memory.}

Every node has a bounded local memory. 
In this work we focus on algorithms with a ``small'' memory, i.e., logarithmic size memory. In particular, 
this work focuses on
algorithms
that could store ids of only a constant number of links (e.g., leading to some of its neighbors).
The local memory can store enough information so that the node only knows which of its neighbors has sent a request in the current round, 
and it also has enough space to save the rumor. 

\paragraph{Local steps: Sending, delivering and reading a message.}

Each node runs a given algorithm in consecutive steps.
Every node can read one message per time step from the buffer. 
If the message contains the rumor, the node stores the rumor provided it has not stored the rumor earlier. Otherwise, the rumor is discarded. 
Every node can send one message per time step. 

\paragraph{\push, \pull and \pp spreading algorithms.}\ 

\push algorithm works as follows: if a node has the rumor stored in its memory, it sends a message with the rumor to one neighbor selected uniformly at random among all its neighbors (so called \push action). 
If a node does not have the rumor stored in its memory, it waits until it receives the rumor (in the meantime, it keeps reading messages from its buffer, one reading attempt per step). 

\pull algorithm works symmetrically: If a node does not have the rumor stored in its memory, it sends a 
request to a neighbor selected uniformly at random among all its neighbors (so called \pull request). 
If a node has the rumor stored in its memory, it keeps reading messages from its buffer and answering every \pull request it reads at the same step it reads the request. 
We recall that 
a node can read only one message from its buffer; Hence, even if the request from some node is in its buffer, it does not mean that it will be answered immediately.  

\pp algorithm is a combination of both \push and \pull algorithms. If a node has the rumor stored in its memory, it acts as follows: If it has no pending request, it selects one neighbor uniformly at random among all its neighbors and sends the rumor to that neighbor (\push action). If it has pending requests in memory, it sends the rumor to the 
currently read \pull request from its buffer, if any (answering \pull request).
We describe \pp algorithm for the sake of completeness, though in this document we focus only on the study of \push and \pull algorithms in the newly introduced model with buffers. 

\paragraph{Rumor spreading problem.}

In the beginning of an execution, there is one node that has a rumor; we call that node the 
{\em source node}. 
The goal is to spread the rumor to every node in the network. 

In the execution of an algorithm, a node could be in one of the two states: informed and uninformed.
We say that a node is {\em informed} at a step 
if it is the source node or it has already read a message containing the rumor by the step.
A node that is not yet informed is called {\em uninformed}. 
An uninformed node that has a message containing the source rumor in its buffer
is called {\em nearly-informed}.
Note that a node can locally recognize and remember whether it is informed or not,
but an uninformed node cannot locally check if it is nearly-informed, as this would require
reading messages containing the source rumor from its buffer, 
which will eventually take place, but not faster than one message
per step. 

\paragraph{Objective function.}
Let us denote by $\mathcal{I}^{\mathcal{A}}_t$ the set of informed nodes in an execution of spreading algorithm $\mathcal{A}$ at step $t$. 
The \emph{time complexity of an execution} of spreading algorithm $\mathcal{A}$  is defined as the expected time step $t$ such that $\mathcal{I}^\mathcal{A}_t = V$.
Time complexity for an execution of a spreading algorithm depends on the source node. Hence, we define the 
 \emph{time complexity of spreading algorithm $\mathcal{A}$} 
 for rumor spreading in a graph $G$, denoted by $T^\mathcal{A}(G)$, as the worst time complexity of any execution of spreading algorithm $\mathcal{A}$ in $G$ over all possible source nodes.

%% file: push.tex
\section{\push algorithm: equivalence between classical and buffer models}\label{sec:push}


The main goal of this section is to show that when nodes in a network with buffers 
spread the rumor using \push algorithm, its time complexity is 
equivalent to the time complexity in the classical rumor spreading model 
without buffers.  
The intuition behind it is the fact that, if nodes use only \push algorithm to transmit the rumor, 
uninformed nodes do not transmit any type of messages that slow down (or speed up) 
the whole process by creating queues in the buffers. They wait until they receive the rumor to then start to transmit it.  

\begin{theorem}\label{thm:pushclass=pushbuff}
Let $G$ be a graph and $u \in V$ the source node. Let $\mathcal{E}$ be an execution of \push algorithm that spreads the rumor to all the nodes in $G$ in time $T$ in the classical model. Then, when the model with buffers is considered for the same graph $G$ and the same initial node $u$, the execution $\mathcal{E}$ of the \push algorithm has the same time complexity $T$.
\end{theorem}

\begin{proof}
The proof is by induction. We prove that for all time step $t$, the set of informed nodes in the classical model $\mathcal{I}_t$ is equal to the set of informed nodes in the model with buffers $\mathcal{I}'_t$. 
The first step of the induction holds by the assumption of the Theorem, i.e., $u \in V$ is the node that initially holds the rumor in both cases, the classical model and the model with buffers. Hence, it holds $\{u\}  = \mathcal{I}_0 = \mathcal{I}'_0$. 

Let us assume now that  $\mathcal{I}_t = \mathcal{I}'_t$ for some time step $t$. We will prove that $\mathcal{I}_{t+1}= \mathcal{I}'_{t+1}$. If the set of informed nodes does not change from step $t$ to step $t+1$ in the classical model, it means that in the execution $\mathcal{E}$ every random selection did not select an uninformed node. Hence, since the model with buffers is executing the same execution $\mathcal{E}$ and $\mathcal{I}_t = \mathcal{I}'_t$, it holds that $\mathcal{I}'_{t+1}$ does not differ from $\mathcal{I}'_t$. Therefore, it holds that $\mathcal{I}_{t+1}=\mathcal{I}_t  = \mathcal{I}'_t = \mathcal{I}'_{t+1}$. On the other hand, if some uninformed nodes become informed at time step $t+1$ in the classical model, since the model with buffers is executing the same execution $\mathcal{E}$, then the same nodes receive the rumor in their buffers at the same time step $t+1$. Now, since those nodes were uninformed and the spreading algorithm is \push, the buffers of these nodes did not contain any message (otherwise, the message would have been the rumor and the nodes would have been informed). 
Therefore, these nodes read the rumor at the same step $t+1$ and become informed at the same step $t+1$. In conclusion, $\mathcal{I}_{t+1}= \mathcal{I}'_{t+1}$.  
\qed
\end{proof}

Theorem \ref{thm:pushclass=pushbuff} allows us to state the following corollary of equivalence between classical model and model with buffers when \push is the algorithm used to spread the rumor. 
\begin{corollary}
Every result that holds for the time complexity of rumor spreading in the classical model for \push algorithm also holds for the model with buffers when \push algorithm is used to spread the rumor.
\end{corollary}

%% file: pull.tex
\section{\pull protocol: bounds for the time complexity in the buffer model}\label{sec:pull}


In this Section we give upper bounds on the time complexity of \pull spreading algorithm 
in the case of regular graphs and general graphs.

In order to mark clear differences between the performance of \pull algorithm in the classical model and in the newly introduce model with buffers, we first describe the following example. 
Consider a graph that consists in the concatenation of $d$ $\delta$-stars, where $\delta$-star is a central node connected to $\delta$ leaves. 
More precisely, the concatenation of $d$ $\delta$-stars is the graph where $d$ 
$\delta$-stars are connected by $d-1$ additional edges such that the central nodes form a path of length $d$. 
When the rumor is spread in such a graph via \pull according to the classical model, any informed central node informs in constant number of rounds all its leaves. 
On the other hand, by the coupon collector argument,  
the rumor is transmitted from one central node to the next central node in $O(\delta\log\delta)$ rounds.  Therefore, when the source node is the leftmost central node, $O(d\delta\log\delta)$ rounds are enough to inform the rightmost central node and all the leaves.

If the same situation is analyzed under the model with buffers, we obtain a completely different result. Since leaves have only one neighbor, which is its corresponding central node, they send one request message to the central node in every round while being uninformed. Hence, each central node collects $\delta$ request messages per each round when all its leaves are uninformed. In the best case, the rumor will be transmitted from one central node to the next central node in the round right after the former becomes informed. Even though, the second central node has accumulated $\delta$ more requests in its buffer in that round, which will delay the process of informing the second central node by at least $\delta$ more rounds. 
Therefore, in the best case, when the source node is the leftmost central node, the rightmost central node will become informed after $\Omega(\delta^d)$ steps. This shows a huge gap between the performance of \pull analyzed in the classical model and in the model with buffers, for the same graph.

\subsection{Upper bounds}
For any network $G$, 
we denote by $d_u$ the degree of node $u$ and by $\Delta_{max}$ the maximum degree in the network.
On the other hand, we denote by $E_u$ the expected number of messages that could be received by node $u$ in a single step provided that each of its neighbors transmits a message to one of its own neighbors (randomly selected). 
Let us denote by $E_{max}$ the maximum of such expectations over all nodes. 
\begin{lemma}
Let $G$ be any graph. It holds that  $E_{max} \geq 1$.  
\end{lemma}
\begin{proof}
First, we point out the fact that in any graph, it holds the following equality: $\sum_{u \in V} E_u = |V|$. This equality comes from the following. By definition of the expectation, it holds: $E_u = \sum_{u \in N(u)}1/d_u$. Then, in the sum 
$\sum_{u \in V} E_u$ the inverse of the degree of every node appears summed one time per each neighbor of it, which sum up to $1$ for each node. 

Now, let us assume by contradiction that there exists a graph such that $E_{max} < 1$. By definition of $E_{max}$, it holds  $\sum_{u \in V} E_u \leq n E_{max}$. By the assumption,  it holds $\sum_{u \in V} E_u \leq n E_{max} < n$. 
Hence, we obtain a contradiction. \qed
\end{proof}
According to the previous Lemma, 
a given graph is either regular or its corresponding $E_{max}$ is strictly larger than $1$. 

In order to prove the upper bound, we use the following notation.
Let us denote by $T_i$ the expected number of 
steps
such that every node at distance $i$ from the source is informed or nearly-informed.
%
%
\begin{lemma}\label{lem:recursive-formula}
The following recursive formula holds: 
\[
T_{i+1} \le
\left\{
\begin{array}{ll}
T_i + O(\sqrt{T_i} + \Delta_{max}\log\Delta_{max}) & \mbox{ if }  E_{u} -1=0 \mbox{ for all } u \in V  \\ 
T_i (E_{max} - 1) + O(\Delta_{max}\log\Delta_{max}) & \mbox{ if } E_{max}-1 \mbox{ is positive}
\end{array}
\right.
\]
\end{lemma}

\begin{proof}
We prove the Theorem by upper bounding the expected time complexity $T_i$ that every node in the $i$-th layer of the BFS tree rooted at the source is informed or nearly-informed (i.e., has the rumor in its buffer). 
Note that the number of messages stored in the buffer of every node may grow in time. 
Nevertheless, once the rumor has been sent to a node (and thus was put to the buffer of that node), 
no message arriving later to the buffer 
could delay the time of reading the rumor by the node from its buffer. 
For every node, we bound the number of messages in the buffer that arrive earlier 
than the rumor by taking into account the layer number of the node. 


In case $E_{max}-1=0$, in expectation there is at most one message arriving per to each node. 
Though, due to the model, if the buffer is not empty, the node processes one message from the buffer at each step. Let us denote by $E(T)$ the expected number of messages queued in the buffer of a node after $T$ steps. 
It is known that the probability of the buffer size being $O(\sqrt{T})$ is constant \cite{petrov1975sums}. 
It follows directly from the fact that the increase of the buffer size in each step can be modeled
by $\Delta$ Bernoulli trials with probability of success $1/\Delta$ minus one,
which in the period of $T$ steps is $O(\sqrt{T})$ with a constant probability (e.g., by using
Chernoff bound). 

Hence, the following recursive formula holds: $T_{i+1} \le T_i + O(\sqrt{T_i}) + O(\Delta\log\Delta)$. 
The $T_i$ part in this formula represents the time required so that every node in the previous layer 
is informed or nearly-informed.
The $O(\sqrt{T_i})$ part in the formula is an upper bound on the maximum size of the buffer at that step,
so that nearly-informed nodes could become informed after reading messages
from their buffers one after another. 
Finally, $O(\Delta\log\Delta)$ is the number of pull requests that any node (in the previous layer) has to answer 
so that each of its neighbors, and in particular - neighbors in the next layer, receives the rumor in its buffer;
this bound follows by the coupon collector argument.


In case $E_{max}-1<0$ the above mentioned expected buffer size is constant.
Hence, the following recursive formula holds: $T_{i+1} \le T_i + O(\Delta\log\Delta)$. 
The $T_i$ part in this formula represents the time required so that every node in the previous layer 
is informed or nearly-informed. Moreover, the $O(\Delta\log\Delta)$ is the number of pull requests that any node (in the previous layer) has to answer 
so that each neighbors in the next layer receives the rumor in its buffer

Finally, in the case when $E_{max}-1>0$ the above mentioned expected buffer size grows at most by a factor of $E_{max}-1$ at each step. 
Then, it holds that $T_{i+1} \leq T_i (E_{max} - 1) + O(\Delta_{max}\log\Delta_{max})$. 
 %
\qed
\end{proof}

\begin{theorem}
\label{t:regular-upper}
The expected number of steps for rumor spreading,
taken by protocol \pull in any synchronous execution on 
any regular networks
of diameter $D$ and degree $\Delta$, is 
$O(D^2 \Delta \log \Delta)$. 
\end{theorem}
\begin{proof}
%


Since the graph is regular, in expectation there is one message arriving per step. 
Therefore, it applies the case $E_{max}-1=0$ in the previous Lemma. 
%
%
Hence, the following recursive formula holds: $T_{i+1} \le T_i + O(\sqrt{T_i}) + O(\Delta\log\Delta)$. 

We set $T_1 = O(\Delta \log \Delta)$.
Indeed, by applying coupon collector argument, by time $\Theta(\Delta \log \Delta)$, in expectation, 
each neighbor of the source will receive the rumor. 
%
%
%
Standard algebraic calculation, c.f., \cite{WolframMathematica}, shows that the solution to the above defined recursive formula is 
$T_{D} = O(D^2\Delta \log \Delta)$.
The expected number of steps after $T_D$ during which nodes nearly-informed at time $T_D$
become informed is $O(\sqrt{T_D})$.
Therefore, the expected number of steps 
so that the rumor is read by all nodes is 
$O(D^2 \Delta \log \Delta)$.
\qed
\end{proof}

\begin{theorem}
\label{t:general-upper}
The expected number of steps for rumor spreading,
taken by protocol \pull in any synchronous execution on 
any  network
of diameter $D$ and maximum degree $\Delta_{max}$, is 
$O(\Delta_{max}\log\Delta_{max}(E_{max} - 1)^D)$. 
\end{theorem}
\begin{proof}
In this case, it holds the second case of Lemma \ref{lem:recursive-formula}. Hence, the corresponding recursive formula is $T_{i+1} \le T_i (E_{max} - 1) + O(\Delta_{max}\log\Delta_{max})$. 
We set $T_1 = O(\Delta \log \Delta)$ and by algebraic manipulation we obtain that $T_{D} = O(\Delta_{max}\log\Delta_{max}(E_{max} - 1)^D)$. \qed
\end{proof}

Finally, let us recall the example that opens this Section. In that case, we showed that for the concatenation of $d$ $\delta$-stars \pull algorithm requires at least $\Omega(\delta^D)$ steps to spread the rumor. Note that the upper bound given by Theorem~\ref{t:general-upper} in that case is $O(\delta^{D+1}\log\delta)$.

%% file: conclusions.tex
\section{Conclusions}\label{sec:conclusions}

In this paper we formally introduced a new model for randomized rumor spreading,
which encapsulates realistic behavior of nodes restricted by one operation on the 
message buffer at a time. We demonstrated differences between our model and the
classical model without buffering, which is particularly visible in some non-regular
types of graphs. These kinds of graphs
might be important from practical perspective, as they model hierarchical
topologies with dense end points that could be used in some future technologies
such as the Internet of Things.
Obtaining results for other specific classes of networks and communication tasks is the
most important open direction.

The new model is also very challenging from mathematical perspective, as it 
combines in a non-trivial way different types of random processes, such as message queuing
and rumor spreading. Therefore, tight estimates of rumor spreading time for specific
classes of networks might be much more difficult to obtain than in the classical model. 
New generic methods of transferring results between the two models are also of great
importance.

Finally, considering other variations of the basic protocols, especially \pull but also \pp,
and analyzing their performance in synchronous and asynchronous executions, is another
perspective direction.

%% file: main-buffer-model.bbl
\begin{thebibliography}{10}
	
	\bibitem{1638541}
	Boyd, S., Ghosh, A., Prabhakar, B., Shah, D.:
	\newblock Randomized gossip algorithms.
	\newblock Information Theory, IEEE Transactions on \textbf{52}(6) (june 2006)
	2508 -- 2530
	
	\bibitem{DBLP:conf/stoc/Censor-HillelHKM12}
	Censor-Hillel, K., Haeupler, B., Kelner, J.A., Maymounkov, P.:
	\newblock Global computation in a poorly connected world: fast rumor spreading
	with no dependence on conductance.
	\newblock In: Proceedings of the 44rd annual ACM symposium on Theory of
	computing. STOC '12, ACM (2012)  961--970
	
	\bibitem{DBLP:conf/soda/Censor-HillelS11}
	Censor-Hillel, K., Shachnai, H.:
	\newblock Fast information spreading in graphs with large weak conductance.
	\newblock In Randall, D., ed.: Proceedings of the Twenty-Second Annual ACM-SIAM
	Symposium on Discrete Algorithms, SODA 2011, San Francisco, California, USA,
	January 23-25, 2011. SODA'11, SIAM (2011)  440--448
	
	\bibitem{DBLP:conf/stoc/ChierichettiLP10}
	Chierichetti, F., Lattanzi, S., Panconesi, A.:
	\newblock Almost tight bounds for rumour spreading with conductance.
	\newblock In Schulman, L.J., ed.: Proceedings of the 42nd ACM Symposium on
	Theory of Computing, STOC 2010, Cambridge, Massachusetts, USA, 5-8 June 2010.
	STOC, ACM (2010)  399--408
	
	\bibitem{DBLP:conf/soda/ChierichettiLP10}
	Chierichetti, F., Lattanzi, S., Panconesi, A.:
	\newblock Rumour spreading and graph conductance.
	\newblock In Charikar, M., ed.: Proceedings of the Twenty-First Annual ACM-SIAM
	Symposium on Discrete Algorithms, SODA 2010, Austin, Texas, USA, January
	17-19, 2010. SODA'10, SIAM (2010)  1657--1663
	
	\bibitem{DBLP:journals/tcs/ChierichettiLP11}
	Chierichetti, F., Lattanzi, S., Panconesi, A.:
	\newblock Rumor spreading in social networks.
	\newblock Theor. Comput. Sci. \textbf{412}(24) (2011)  2602--2610
	
	\bibitem{DBLP:journals/dam/ChlebusDP94}
	Chlebus, B.S., Diks, K., Pelc, A.:
	\newblock Fast gossiping with short unreliable messages.
	\newblock Discrete Applied Mathematics \textbf{53}(1-3) (1994)  15--24
	
	\bibitem{Demers1988}
	Demers, A., Greene, D., Houser, C., Irish, W., Larson, J., Shenker, S.,
	Sturgis, H., Swinehart, D., Terry, D.:
	\newblock Epidemic algorithms for replicated database maintenance.
	\newblock SIGOPS Oper. Syst. Rev. \textbf{22} (January 1988)  8--32
	
	\bibitem{DBLP:conf/icalp/DoerrF11}
	Doerr, B., Fouz, M.:
	\newblock Asymptotically optimal randomized rumor spreading.
	\newblock In Aceto, L., Henzinger, M., Sgall, J., eds.: Automata, Languages and
	Programming - 38th International Colloquium, ICALP 2011, Zurich, Switzerland,
	July 4-8, 2011, Proceedings, Part II. Volume 6756 of Lecture Notes in
	Computer Science., Springer (2011)  502--513
	
	\bibitem{Elsasser:2006:CCR:1148109.1148135}
	Els\"{a}sser, R.:
	\newblock On the communication complexity of randomized broadcasting in
	random-like graphs.
	\newblock In: Proceedings of the eighteenth annual ACM symposium on Parallelism
	in algorithms and architectures. SPAA '06 (2006)  148--157
	
	\bibitem{DBLP:conf/stacs/ElsasserS07}
	Els{\"a}sser, R., Sauerwald, T.:
	\newblock Broadcasting vs. mixing and information dissemination on cayley
	graphs.
	\newblock In Thomas, W., Weil, P., eds.: STACS. Volume 4393 of Lecture Notes in
	Computer Science., Springer (2007)  163--174
	
	\bibitem{DBLP:journals/tcs/ElsasserS09}
	Els{\"a}sser, R., Sauerwald, T.:
	\newblock On the runtime and robustness of randomized broadcasting.
	\newblock Theor. Comput. Sci. \textbf{410}(36) (2009)  3414--3427
	
	\bibitem{DBLP:journals/rsa/FeigePRU90}
	Feige, U., Peleg, D., Raghavan, P., Upfal, E.:
	\newblock Randomized broadcast in networks.
	\newblock Random Struct. Algorithms \textbf{1}(4) (1990)  447--460
	
	\bibitem{fountoulakis2012rumor}
	Fountoulakis, N., Panagiotou, K.:
	\newblock Rumor spreading on random regular graphs and expanders.
	\newblock Random Structures \& Algorithms (2012)
	
	\bibitem{frieze1985shortest}
	Frieze, A.M., Grimmett, G.R.:
	\newblock The shortest-path problem for graphs with random arc-lengths.
	\newblock Discrete Applied Mathematics \textbf{10}(1) (1985)  57--77
	
	\bibitem{DBLP:conf/stacs/Giakkoupis11}
	Giakkoupis, G.:
	\newblock Tight bounds for rumor spreading in graphs of a given conductance.
	\newblock In Schwentick, T., D{\"u}rr, C., eds.: STACS. Volume~9 of LIPIcs.,
	Schloss Dagstuhl - Leibniz-Zentrum fuer Informatik (2011)  57--68
	
	\bibitem{DBLP:conf/soda/GiakkoupisS12}
	Giakkoupis, G., Sauerwald, T.:
	\newblock Rumor spreading and vertex expansion.
	\newblock In Rabani, Y., ed.: Proceedings of the Twenty-Third Annual ACM-SIAM
	Symposium on Discrete Algorithms, SODA 2012, Kyoto, Japan, January 17-19,
	2012. SODA'12, SIAM (2012)  1623--1641
	
	\bibitem{Harchol-Balter1999}
	Harchol-Balter, M., Leighton, T., Lewin, D.:
	\newblock Resource discovery in distributed networks.
	\newblock In: Proceedings of the eighteenth annual ACM symposium on Principles
	of distributed computing. PODC '99, ACM (1999)  229--237
	
	\bibitem{Karp:2000:RRS:795666.796561}
	Karp, R., Schindelhauer, C., Shenker, S., Vocking, B.:
	\newblock Randomized rumor spreading.
	\newblock In: Proceedings of the 41st Annual Symposium on Foundations of
	Computer Science. (2000)  565--574
	
	\bibitem{DBLP:journals/siamcomp/KrummeCV92}
	Krumme, D.W., Cybenko, G., Venkataraman, K.N.:
	\newblock Gossiping in minimal time.
	\newblock SIAM J. Comput. \textbf{21}(1) (1992)  111--139
	
	\bibitem{WolframMathematica}
	{M}athematica 9, W.:
	\newblock Produced by {W}olfram {R}esearch, {I}nc.
	\newblock {\tt http://www.wolfram.com/mathematica/}
	
	\bibitem{petrov1975sums}
	Petrov, V.:
	\newblock Sums of independent random variables.
	\newblock Ergebnisse der Mathematik und ihrer Grenzgebiete. Springer-Verlag
	(1975)
	
	\bibitem{Pittel:1987:SR:37387.37400}
	Pittel, B.:
	\newblock On spreading a rumor.
	\newblock SIAM J. Appl. Math. \textbf{47} (1987)  213--223
	
	\bibitem{springerlink:10.1007/s00453-008-9245-4}
	Sauerwald, T.:
	\newblock On mixing and edge expansion properties in randomized broadcasting.
	\newblock Algorithmica \textbf{56} (2010)  51--88
	
	\bibitem{DBLP:conf/soda/SauerwaldS11}
	Sauerwald, T., Stauffer, A.:
	\newblock Rumor spreading and vertex expansion on regular graphs.
	\newblock In Randall, D., ed.: Proceedings of the Twenty-Second Annual ACM-SIAM
	Symposium on Discrete Algorithms, SODA 2011, San Francisco, California, USA,
	January 23-25, 2011. SODA'11, SIAM (2011)  462--475
	
	\bibitem{vanRenesse:2009:GFD:1659232.1659238}
	van Renesse, R., Minsky, Y., Hayden, M.:
	\newblock A gossip-style failure detection service.
	\newblock In: Proceedings of the IFIP International Conference on Distributed
	Systems Platforms and Open Distributed Processing. Middleware '98,
	Springer-Verlag (1998)  55--70
	
\end{thebibliography}
